\documentclass[12pt]{article}

\usepackage[margin=1in]{geometry}
\usepackage{CJK, amsmath, amssymb, amsthm, authblk, hyperref, tikz}

\newtheorem{lemma}{Lemma}
\newtheorem{theorem}{Theorem}
\newtheorem{corollary}{Corollary}

\theoremstyle{definition}
\newtheorem{definition}{Definition}

\DeclareMathOperator{\tr}{tr}
\DeclareMathOperator*{\E}{\mathbb E}
\DeclareMathOperator{\diag}{diag}

\begin{document}

\title{Extensive entropy from unitary evolution}

\begin{CJK*}{UTF8}{}

\CJKfamily{gbsn}

\author{Yichen Huang (黄溢辰)\thanks{yichuang@mit.edu}}
\affil{Center for Theoretical Physics, Massachusetts Institute of Technology, Cambridge, Massachusetts 02139, USA\thanks{Work from home in Commerce Township, Michigan 48390, USA, where winter is cold and snowy.}}

\maketitle

\end{CJK*}

\begin{abstract}

In quantum many-body systems, a Hamiltonian is called an ``extensive entropy generator'' if starting from a random product state the entanglement entropy obeys a volume law at long times with overwhelming probability. We prove that (i) any Hamiltonian whose spectrum has non-degenerate gaps is an extensive entropy generator; (ii) in the space of (geometrically) local Hamiltonians, the non-degenerate gap condition is satisfied almost everywhere. Specializing to many-body localized systems, these results imply the observation stated in the title of Bardarson et al. [PRL 109, 017202 (2012)].\footnote{I first met my doctoral advisor Joel E. Moore in December 2011, when he was working with his group members on the project \cite{BPM12}. More than nine years later, I was very fortunate to find a proof of the observation stated in the title of Ref. \cite{BPM12}. This reminds me of the sunny summer and rainy (not snowy) winter in Berkeley. I am very grateful to Joel for his guidance and support since we first met.}

\end{abstract}

\section{Introduction}

Entropy is a fundamental concept in thermodynamics and statistical mechanics. In textbooks and introductory courses, we learned that entropy is an extensive quantity:
\begin{quote}
Initializing a system in a low entropy state, generically (although not always) the entropy will grow with time and eventually become proportional to the system size.
\end{quote}
The main contribution of this paper is to provide a mathematical characterization of this statement in quantum many-body systems. Our rigorous results not only show how extensive entropy emerges from the unitary evolution under a (geometrically) local Hamiltonian, but also demonstrate the genericness of extensive entropy.

Suppose an isolated quantum spin system is initialized in a pure state. Under unitary evolution the system stays in a pure state and hence its entropy is always zero. To observe non-trivial entropy dynamics, we divide the system into two parts $A$ and $B$. Assume without loss of generality that subsystem $A$ is smaller than or equal to half the system size. We view $B$ as a bath of $A$ and consider the entropy of $A$. This entropy is called entanglement entropy, and extensive subsystem entropy is also known as a volume law for entanglement. The most general unentangled state with zero entropy for all subsystems is a random product state, where each spin is chosen independently and uniformly at random on the Bloch sphere.

Our main result is the following. A Hamiltonian is called an ``extensive entropy generator'' if starting from a random product state the entanglement entropy obeys a volume law at long times with overwhelming probability. We prove that (i) any Hamiltonian whose spectrum has non-degenerate gaps is an extensive entropy generator; (ii) in the space of Hamiltonians with short-range interactions, the non-degenerate gap condition is satisfied almost everywhere.

As a byproduct, we solve a mathematical problem in many-body localization. It is well known that in one-dimensional Anderson localized systems, starting from a random product state the entanglement entropy remains bounded at all times \cite{ANSS16}. However, upon adding a generic local perturbation the system becomes many-body localized (MBL), and unbounded growth of entanglement was observed numerically \cite{ZPP08, BPM12, NKH14} and then explained heuristically \cite{VA13, SPA13U, HNO14}. We prove this observation.

\section{Preliminaries}

Throughout this paper, standard asymptotic notations are used extensively. Let $f,g:\mathbb R^+\to\mathbb R^+$ be two functions. One writes $f(x)=O(g(x))$ if and only if there exist constants $M,x_0>0$ such that $f(x)\le Mg(x)$ for all $x>x_0$; $f(x)=\Omega(g(x))$ if and only if there exist constants $M,x_0>0$ such that $f(x)\ge Mg(x)$ for all $x>x_0$.

\begin{definition} [entanglement entropy] \label{def:ee}
The entanglement entropy of a bipartite pure state $\rho_{AB}$ is defined as the von Neumann entropy
\begin{equation}
S(\rho_A):=-\tr(\rho_A\ln\rho_A)
\end{equation}
of the reduced density matrix $\rho_A=\tr_B\rho_{AB}$.
\end{definition}

Consider a system of $N$ spins or qu\emph{d}its with local dimension $d$ so that the dimension of the total Hilbert space is $d^N$.

\begin{definition} [Haar-random product state] \label{def:haar}
Let $|\Psi\rangle=\bigotimes_{j=1}^N|\Psi_j\rangle$ be a Haar-random product state, where each $|\Psi_j\rangle$ is chosen independently and uniformly at random with respect to the Haar measure.
\end{definition}

\begin{definition} [non-degenerate spectrum] \label{def:nds}
The spectrum of a Hamiltonian is non-degenerate if all eigenvalues are distinct.
\end{definition}

\begin{definition} [non-degenerate gaps] \label{def:ndg}
The spectrum $\{E_j\}$ of a Hamiltonian has non-degenerate gaps if the differences $\{E_j-E_k\}_{j\neq k}$ are all distinct, i.e., for any $j\neq k$,
\begin{equation} \label{eq:ndg}
E_j-E_k=E_{j'}-E_{k'}\implies(j=j')~{\rm and}~(k=k').
\end{equation}
\end{definition}

By definition, the non-degenerate gap condition implies that the spectrum is non-degenerate.

\section{Results}

Let $m,n$ be positive integers such that $mn$ is a multiple of the system size $N$. Let $A_1,A_2,\ldots,A_m$ be $m$ possibly overlapping subsystems, each of which consists of exactly $n\le N/2$ spins. Suppose that each spin in the system is in exactly $mn/N$ out of these $m$ subsystems. For each $j$, let $\bar A_j$ be the complement of $A_j$ so that $A_j\otimes\bar A_j$ defines a bipartition of the system.

\begin{theorem} \label{thm:low}
Initialize the system in a Haar-random product state $|\Psi\rangle$ (Definition \ref{def:haar}). Let
\begin{equation}
\rho_{A_j}(t):=\tr_{\bar A_j}\rho(t),\quad\rho(t):=e^{-iHt}|\Psi\rangle\langle\Psi|e^{iHt}
\end{equation}
be the reduced density matrix of subsystem $A_j$ at time $t$. We consider the scenario that
\begin{equation} \label{eq:upto}
n\le\frac{(1-\epsilon)N\ln\frac{d+1}{2}}{2\ln d},
\end{equation}
where $\epsilon>0$ is an arbitrarily small constant. For any Hamiltonian $H$ whose spectrum has non-degenerate gaps (Definition \ref{def:ndg}),
\begin{equation} \label{eq:vol}
    \Pr_\Psi\left(\Pr_{t\in\mathbb R}\left(\frac{1}{m}\sum_{j=1}^mS\big(\rho_{A_j}(t)\big)\ge(1-\epsilon)n\sum_{k=2}^d\frac1k\right)=1-e^{-\Omega(N)}\right)=1-e^{-\Omega(N)}.
\end{equation}
\end{theorem}

There is no underlying lattice structure in the statement of Theorem \ref{thm:low}. From now on we focus on quantum lattice systems, which are of particular interest. Without loss of generality, we consider a chain of $N$ spins. It is straightforward to extend Corollary \ref{cor:large} below to higher spatial dimensions.

Let $A$ be a contiguous region of $n$ spins, and $\bar A$ be the rest of the system. Let $\E_{|A|=n}$ denote averaging over all contiguous subsystems of size $n$. With periodic boundary conditions, there are $N$ such subsystems.

\begin{corollary} \label{cor:large}
Initialize the system in a Haar-random product state $|\Psi\rangle$. Let
\begin{equation}
\rho_A(t):=\tr_{\bar A}\rho(t),\quad\rho(t):=e^{-iHt}|\Psi\rangle\langle\Psi|e^{iHt}
\end{equation}
be the reduced density matrix of subsystem $A$ at time $t$. We consider the scenario that
\begin{equation}
\frac{(1-\epsilon)N\ln\frac{d+1}{2}}{2\ln d}<n\le N/2,
\end{equation}
where $\epsilon>0$ is an arbitrarily small constant. For any (not necessarily local) Hamiltonian $H$ whose spectrum has non-degenerate gaps,
\begin{equation}
    \Pr_\Psi\left(\Pr_{t\in\mathbb R}\left(\E_{|A|=n}S\big(\rho_A(t)\big)\ge\left(\frac{1}{2}-\epsilon\right)\frac{N\ln\frac{d+1}{2}}{\ln d}\sum_{k=2}^d\frac1k\right)=1-e^{-\Omega(N)}\right)=1-e^{-\Omega(N)}.
\end{equation}
\end{corollary}

\begin{proof}
The weak monotonicity \cite{LR73} of the von Neumann entropy implies that
\begin{equation}
    \E_{|A|=n}S\big(\rho_A(t)\big)\le\E_{|A|=n+1}S\big(\rho_A(t)\big)
\end{equation}
for any $n<N/2$. Corollary \ref{cor:large} follows from this inequality and Theorem \ref{thm:low}.
\end{proof}

To demonstrate the genericness of non-degenerate gaps, we define ensembles of Hamiltonians with nearest-neighbor interactions in a chain of $N$ spin-$1/2$'s or qubits ($d=2$). We prove that in each ensemble, the set of Hamiltonians whose spectrum has degenerate gaps is of measure zero (Theorem \ref{thm:generic}). Similar results can be proved in a similar way for other types of systems including qudit systems with short-range interactions in higher spatial dimensions or even with non-local interactions (Appendix \ref{app:extend}).

The non-degenerate gap condition (\ref{eq:ndg}) was assumed in many previous works \cite{Rei08, LPSW09, Sho11, SF12, RK12, FBC17, WGRE19} on the equilibration of quantum systems. Both Theorem \ref{thm:low} and the results of these works apply to almost every Hamiltonian in each ensemble we define.

Let
\begin{equation}
\hat\sigma_j^x=
\begin{pmatrix}
0 & 1 \\
1 & 0
\end{pmatrix},\quad
\hat\sigma_j^y=
\begin{pmatrix}
0 & -i \\
i & 0
\end{pmatrix},\quad
\hat\sigma_j^z=
\begin{pmatrix}
1 & 0 \\
0 & -1
\end{pmatrix}
\end{equation}
be the Pauli matrices for the spin-$1/2$ at position $j$, and
\begin{align}
    &J:=(J_k|_{k\in\{x,y,z\}},J_{kl}|_{k,l\in\{x,y,z\}})=(J_x,J_y,J_z,J_{xx},J_{xy},J_{xz},J_{yx},J_{yy},J_{yz},J_{zx},J_{zy},J_{zz}),\\
    &\alpha:=(\alpha_j^k|_{1\le j\le N}^{k\in\{x,y,z\}},\alpha_j^{kl}|_{1\le j\le N-1}^{k,l\in\{x,y,z\}}),\quad R:=(0,1]\times[0,1]\times(0,1]\times[0,1]^{\times8}\times(0,1]\subset\mathbb R^{12}.
\end{align}
Each particular $J$ defines an ensemble of Hamiltonians
\begin{align}
    &H_J:=\{H_J(\alpha):\alpha\in[-1,1]^{\times(12N-9)}\},\\
    &H_J(\alpha):=\sum_{j=1}^N\sum_{k\in\{x,y,z\}}J_k\alpha_j^k\hat\sigma_j^k+\sum_{j=1}^{N-1}\sum_{k,l\in\{x,y,z\}}J_{kl}\alpha_j^{kl}\hat\sigma_j^k\hat\sigma_{j+1}^l. \label{eq:HJa}
\end{align}
Note that $J\in R$ implies $J_x,J_z,J_{zz}>0$. This rules out the possibility that $H_J$ is an ensemble of free-fermion Hamiltonians.

\begin{theorem} \label{thm:generic}
For any $J\in R$, the set of all $\alpha$ such that the spectrum of $H_J(\alpha)$ has degenerate gaps is of measure zero.
\end{theorem}

It is an open problem to prove an analogue of this theorem for translationally invariant systems. Progress in this direction was made in Ref. \cite{HH19}.

\section{Many-body localization}

$H_J$ is MBL for some $J\in R$. In particular, $H_J$ reduces to the Imbrie model \cite{Imb16} when
\begin{equation}
    J=(\lambda,0,1,0,0,0,0,0,0,0,0,1)
\end{equation}
with small $\lambda>0$.

Other models of MBL can be considered. For example, let
\begin{equation}
    h:=(h_x,h_z,h_{zz}),\quad\gamma:=(\gamma_j^k|_{1\le j\le N}^{k\in\{x,z\}},\gamma_j^{zz}|_{1\le j\le N-1}).
\end{equation}
Each particular $h$ defines an ensemble of Hamiltonians
\begin{align}
    &H^{XXZ}_h:=\{H^{XXZ}_h(\gamma):\gamma\in[-1,1]^{\times(3N-1)}\},\\
    &H^{XXZ}_h(\gamma)=\sum_{j=1}^N(h_x\gamma_j^x\hat\sigma_j^x+h_z\gamma_j^z\hat\sigma_j^z)+\sum_{j=1}^{N-1}(\hat\sigma_j^x\hat\sigma_{j+1}^x+\hat\sigma_j^y\hat\sigma_{j+1}^y+h_{zz}\gamma_j^{zz}\hat\sigma_j^z\hat\sigma_{j+1}^z).
\end{align}
Note that $H_h^{XXZ}$ is a perturbed random-field $XX$ chain if $h_x,h_{zz}>0$ are small.

\begin{corollary} \label{cor:mbl}
For any $h\in(0,+\infty)^{\times3}$, the set of all $\gamma$ such that the spectrum of $H^{XXZ}_h(\gamma)$ has degenerate gaps is of measure zero.
\end{corollary}

The result stated in the last paragraph of the introduction follows by combining Theorem \ref{thm:low} and Corollary \ref{cor:mbl}.

Finally, we prove that the volume-law coefficient in Eq. (\ref{eq:vol}) is tight in a particular MBL system. Consider a chain of $N$ qudits labeled by $1,2,\ldots,N$. Let $\hat S_j^z$ be the $z$ component of the spin operator at position $j$. Let $(h_1,h_2,\ldots,h_N)\in\mathbb R^N$ be such that the spectrum of
\begin{equation}
    H^{\rm loc}:=\sum_{j=1}^Nh_j\hat S_j^z
\end{equation}
is non-degenerate. Let $H^{\rm mbl}:=H^{\rm loc}+\Delta H$, where $\Delta H$ is an infinitesimal random local perturbation. Appendix \ref{app:qudit} proves that the spectrum of $H^{\rm mbl}$ almost surely has non-degenerate gaps. Let $A\subset\{1,2,\ldots,N\}$ so that $A\sqcup\bar A$ defines a bipartition of the system.

\begin{theorem} \label{thm:tight}
Initialize the system in a Haar-random product state $|\Psi\rangle$. Let
\begin{equation}
\rho_A(t):=\tr_{\bar A}\rho(t),\quad\rho(t):=e^{-iH^{\rm mbl}t}|\Psi\rangle\langle\Psi|e^{iH^{\rm mbl}t}
\end{equation}
be the reduced density matrix of subsystem $A$ at time $t$. For any $A$,
\begin{equation}
    \E_\Psi\lim_{\tau\to+\infty}\frac1\tau\int_0^\tau S\big(\rho_A(t)\big)\,\mathrm dt\le |A|\sum_{k=2}^d\frac1k.
\end{equation}
\end{theorem}

\section*{Acknowledgments}

I would like to thank my postdoctoral advisor Aram W. Harrow for comments on this paper and for collaboration on a related project \cite{HH19}. This work was supported by NSF grant PHY-1818914 and a Samsung Advanced Institute of Technology Global Research Partnership.

\appendix

\section{Proof of Theorem \ref{thm:low}}

\begin{lemma}
Let $|0\rangle$ be a particular state in $\mathbb C^d$. For a Haar-random state $|\Psi'\rangle$,
\begin{equation}
    M_\alpha:=\E_{\Psi'}\big|\langle0|\Psi'\rangle\big|^{2\alpha}=\prod_{j=1}^{d-1}\frac{j}{j+\alpha}.
\end{equation}
\end{lemma}

\begin{proof}
Calculating a tedious multivariable integral, the probability density function of the random variable $X:=|\langle0|\Psi'\rangle|^2$ is
\begin{equation} \label{eq:fx}
    f(x)=(d-1)(1-x)^{d-2},\quad x\in[0,1].
\end{equation}
Therefore,
\begin{equation}
    M_\alpha=\int_0^1x^\alpha f(x)\,\mathrm dx=\prod_{j=1}^{d-1}\frac{j}{j+\alpha}.
\end{equation}
\end{proof}

\begin{lemma}
Let $|\phi\rangle$ be an arbitrary state in $(\mathbb C^d)^{\otimes N}$. For a Haar-random product state $|\Psi\rangle$ (Definition \ref{def:haar}) and any $\alpha\ge1$,
\begin{equation} \label{eq:chain}
    \E_{\Psi}\big|\langle \phi|\Psi\rangle\big|^{2\alpha}\le M_\alpha^N.
\end{equation}
\end{lemma}

\begin{proof}
Let $A\otimes\bar A$ be a bipartition of the system, where subsystem $A$ consists of a single spin. Let $\{|j\rangle_A\}_{j=0}^{d-1}$ be the computational basis of subsystem $A$ so that
\begin{equation}
    |\phi\rangle=\sum_{j=0}^{d-1}c_j|j\rangle_A\otimes|\phi_j\rangle_{\bar A},\quad\sum_{j=0}^{d-1}|c_j|^2=1.
\end{equation}
Let $|\Psi\rangle=|\Psi_A\rangle\otimes|\Psi_{\bar A}\rangle$, where $|\Psi_A\rangle$ is a Haar-random state in $\mathbb C^d$ and $|\Psi_{\bar A}\rangle$ is a Haar-random product state in $(\mathbb C^d)^{\otimes(N-1)}$. It is easy to see that for any fixed $|\Psi_{\bar A}\rangle$,
\begin{equation}
    \E_{\Psi_A}\big|\langle \phi|\Psi\rangle\big|^{2\alpha}=M_\alpha\left(\sum_{j=0}^{d-1}|c_j|^2\big|\langle\phi_j|\Psi_{\bar A}\rangle\big|^2\right)^\alpha\le M_\alpha\sum_{j=0}^{d-1}|c_j|^2\big|\langle\phi_j|\Psi_{\bar A}\rangle\big|^{2\alpha}.
\end{equation}
Hence,
\begin{equation} \label{eq:step}
    \E_\Psi\big|\langle \phi|\Psi\rangle\big|^{2\alpha}=\E_{\Psi_{\bar A}}\E_{\Psi_A}\big|\langle \phi|\Psi\rangle\big|^{2\alpha}\le M_\alpha\sum_{j=0}^{d-1}|c_j|^2\E_{\Psi_{\bar A}}\big|\langle\phi_j|\Psi_{\bar A}\rangle\big|^{2\alpha}\le M_\alpha\max_{\phi'}\E_{\Psi_{\bar A}}\big|\langle\phi'|\Psi_{\bar A}\rangle\big|^{2\alpha},
\end{equation}
where $|\phi'\rangle$ is a state in $(\mathbb C^d)^{\otimes(N-1)}$. We obtain (\ref{eq:chain}) by iteratively applying (\ref{eq:step}).
\end{proof}

\begin{lemma} \label{l:effdim}
Let $\{|b_1\rangle,|b_2\rangle,\ldots,|b_{d^N}\rangle\}$ be an arbitrary orthonormal basis of the Hilbert space $(\mathbb C^d)^{\otimes N}$. For a Haar-random product state $|\Psi\rangle$ and any $\alpha\ge1$,
\begin{equation}
\E_\Psi\sum_{j=1}^{d^N}\big|\langle b_j|\Psi\rangle\big|^{2\alpha}\le(M_\alpha d)^N.
\end{equation}
Furthermore, for an arbitrarily small constant $\varepsilon>0$,
\begin{equation} \label{eq:mar}
\Pr_\Psi\left(\sum_{j=1}^{d^N}\big|\langle b_j|\Psi\rangle\big|^{2\alpha}\le(M_\alpha d)^{(1-\varepsilon)N}\right)=1-e^{-\Omega(N)}.
\end{equation}
\end{lemma}

\begin{proof}
The probabilistic bound (\ref{eq:mar}) follows from Markov's inequality.
\end{proof}

For $\alpha=2$, Lemma \ref{l:effdim} reduces to Lemma 5 in Ref. \cite{HH19}.

Let $\{|j\rangle\}_{j=1}^{d^N}$ be a complete set of eigenstates of $H$. Assuming that the spectrum of $H$ is non-degenerate, the energy basis $\{|j\rangle\}$ is unambiguously defined. The effective dimension of a state $|\psi\rangle$ is defined as
\begin{equation}
1/D^{\rm eff}_\psi=\sum_{j=1}^{d^N}\big|\langle j|\psi\rangle\big|^4.
\end{equation}

\begin{lemma} \label{l:diag}
Initialize the system in a pure state $|\psi\rangle$. Let
\begin{equation}
\sigma:=\lim_{\tau\to+\infty}\frac1\tau\int_0^\tau\rho(t)\,\mathrm dt,\quad\rho(t):=e^{-iHt}|\psi\rangle\langle\psi|e^{iHt}
\end{equation}
be the infinite time average and $\sigma_{A_j}:=\tr_{\bar A_j}\sigma$ be the reduced density matrix of subsystem $A_j$. For any Hamiltonian $H$ with non-degenerate spectrum and any $\alpha>1$,
\begin{equation} \label{eq:sub}
    \frac{1}{m}\sum_{j=1}^mS(\sigma_{A_j})\ge\frac{n}{N(1-\alpha)}\ln\sum_{j=1}^{d^N}\big|\langle j|\psi\rangle\big|^{2\alpha}.
\end{equation}
\end{lemma}

\begin{proof}
Expanding $|\psi\rangle$ in the energy basis, it is easy to see that
\begin{equation}
    \sigma=\sum_{j=1}^{d^N}p_j|j\rangle\langle j|,\quad p_j:=\big|\langle j|\psi\rangle\big|^2
\end{equation}
is the so-called diagonal ensemble. Using the strong subadditivity \cite{LR73} of the von Neumann entropy and the monotonicity of the R\'enyi entropy,
\begin{equation}
\frac{1}{m}\sum_{j=1}^mS(\sigma_{A_j})\ge\frac{n}{N}S(\sigma)=-\frac{n}{N}\sum_{j=1}^{d^N}p_j\ln p_j\ge\frac{n}{N(1-\alpha)}\ln\sum_{j=1}^{d^N}p_j^\alpha,\quad\forall\alpha>1.
\end{equation}
\end{proof}

Let $\|X\|_1:=\tr\sqrt{X^\dag X}$ denote the trace norm.

\begin{lemma} [\cite{LPSW09, Sho11}] \label{eff}
Using the notation of Lemma \ref{l:diag}, for any Hamiltonian $H$ whose spectrum has non-degenerate gaps,
\begin{equation} \label{eq:eq}
\lim_{\tau\to+\infty}\frac1\tau\int_0^\tau\|\rho_{A_j}(t)-\sigma_{A_j}\|_1\,\mathrm dt\le d^n\big/\sqrt{D^{\rm eff}_\psi}.
\end{equation}
\end{lemma}

\begin{lemma} [continuity of the von Neumann entropy \cite{Fan73, Aud07}] \label{l:cont}
Let $T:=\|\rho-\sigma\|_1/2$ be the trace distance between two density matrices $\rho,\sigma$ on the Hilbert space $\mathbb C^D$. Then,
\begin{equation}
    |S(\rho)-S(\sigma)|\le T\ln(D-1)-T\ln T-(1-T)\ln(1-T).
\end{equation}
\end{lemma}

Since by definition $0\le T\le1$, the right-hand side of this inequality is well defined.

We are ready to prove Theorem \ref{thm:low}. Lemmas \ref{l:effdim} and \ref{l:diag} imply that
\begin{align}
    &\Pr_\Psi\left(D^{\rm eff}_\Psi\ge\left(\frac{d+1}{2}\right)^{(1-\epsilon/2)N}\right)=1-e^{-\Omega(N)},\\
    &\Pr_\Psi\left(\frac{1}{m}\sum_{j=1}^mS(\sigma_{A_j})\ge\frac{(1-\epsilon/3)n}{\alpha-1}\sum_{k=2}^d\ln\left(1+\frac{\alpha-1}k\right)\right)=1-e^{-\Omega(N)},\quad\forall\alpha>1.
\end{align}
Therefore, it suffices to prove that
\begin{equation} \label{eq:suff}
    \Pr_{t\in\mathbb R}\left(\frac{1}{m}\sum_{j=1}^mS\big(\rho_{A_j}(t)\big)\ge (1-\epsilon)n\sum_{k=2}^d\frac1k\right)=1-e^{-\Omega(N)}
\end{equation}
assuming that
\begin{align}
&D^{\rm eff}_\Psi\ge\left(\frac{d+1}{2}\right)^{(1-\epsilon/2)N},\label{eq:eff}\\
&\frac{1}{m}\sum_{j=1}^mS(\sigma_{A_j})\ge(1-2\epsilon/3)n\sum_{k=2}^d\frac1k.\label{eq:low} 
\end{align}
(\ref{eq:upto}), (\ref{eq:eq}), and (\ref{eq:eff}) imply that
\begin{equation}
\lim_{\tau\to+\infty}\frac1\tau\int_0^\tau\frac{1}{m}\sum_{j=1}^m\|\rho_{A_j}(t)-\sigma_{A_j}\|_1\,\mathrm dt\le d^n\big/\sqrt{D^{\rm eff}_\Psi}\le\left(\frac{2}{d+1}\right)^{\epsilon N/4}=e^{-\Omega(N)}.
\end{equation}
Markov's inequality implies that
\begin{equation} \label{eq:close}
    \Pr_{t\in\mathbb R}\left(\frac{1}{m}\sum_{j=1}^m\|\rho_{A_j}(t)-\sigma_{A_j}\|_1=e^{-\Omega(N)}\right)=1-e^{-\Omega(N)}.
\end{equation}
Due to the continuity of the von Neumann entropy (Lemma \ref{l:cont}),
\begin{equation} \label{eq:cont}
    \frac1m\sum_{j=1}^m\|\rho_{A_j}(t)-\sigma_{A_j}\|_1=e^{-\Omega(N)}\implies\frac1m\sum_{j=1}^m\big|S\big(\rho_{A_j}(t)\big)-S(\sigma_{A_j})\big|=e^{-\Omega(N)}.
\end{equation}
Equation (\ref{eq:suff}) follows from (\ref{eq:low}), (\ref{eq:close}), and (\ref{eq:cont}).

\section{Proof of Theorem \ref{thm:generic}} \label{app:generic}

We begin by following the proof of Lemma 8 in Ref. \cite{HH19}. Let $\{E_j\}_{j=1}^{2^N}$ with be the eigenvalues of $H_J(\alpha)$ and
\begin{equation}
G_J(\alpha):=\prod_{((j\neq j')\lor(k\neq k'))\land((j\neq k')\lor(k\neq j'))}(E_j+E_k-E_{j'}-E_{k'})
\end{equation}
so that $G_J(\alpha)=0$ if and only if the spectrum of $H_J(\alpha)$ has degenerate gaps (Definition \ref{def:ndg}). It is easy to see that $G_J(\alpha)$ is a symmetric polynomial in $E_1,E_2,\ldots,E_{2^N}$. The fundamental theorem of symmetric polynomials implies that $G_J(\alpha)$ can be expressed as a polynomial in $F_1,F_2,\ldots$, where
\begin{equation}
F_k:=\sum_{j=1}^{2^N}E_j^k=\tr\big(H_J^k(\alpha)\big).
\end{equation}
Expanding $H_J^k(\alpha)$ in the Pauli basis and taking the trace, we see that $F_k$ and hence $G_J(\alpha):[-1,1]^{\times(12N-9)}\to\mathbb R$ are polynomials in the entries of $\alpha$. Since the zeros of a multivariable polynomial are of measure zero unless the polynomial is identically zero, it suffices to find a particular $\alpha\in\mathbb R^{12N-9}$ such that the spectrum of $H_J(\alpha)$ has non-degenerate gaps. This is done in the following lemma.

\begin{lemma} \label{l:HN}
For any positive integer $N$, there exists $(h_1,J_1,h_2,J_2,\ldots,h_{N-1},J_{N-1})\in\mathbb R^{2(N-1)}$ such that the spectrum of 
\begin{equation} \label{eq:HN}
    H_N=\hat\sigma^z_1+\sum_{j=1}^{N-1}\big(h_j\hat\sigma^x_{j+1}+J_j(2+\hat\sigma^z_j)\hat\sigma^z_{j+1}\big)
\end{equation}
has non-degenerate gaps.
\end{lemma}

\begin{proof}
We prove by induction on $N$.

{\bf Induction hypotheses.} We have two induction hypotheses:
\begin{enumerate}
    \item The spectrum of $H_N$ has non-degenerate gaps.
    \item $\{\langle j|\hat\sigma^z_N|j\rangle\}$ is a set of pairwise distinct numbers, where $\{|j\rangle\}_{j=1}^{2^N}$ are eigenstates of $H_N$.
\end{enumerate}

{\bf Base case.} The base case $N=1$ is trivially true.

{\bf Induction step.} Assuming the induction hypotheses for $H_N$, we prove those for $H_{N+1}$ using perturbation theory. The unperturbed Hamiltonian and the perturbation are
\begin{equation}
    H_{\rm unp}:=H_N+h_N\hat\sigma^x_{N+1},\quad H_{\rm per}:=J_N(2+\hat\sigma^z_N)\hat\sigma^z_{N+1},
\end{equation}
respectively. Let $\{|j\rangle\}$ be eigenstates of $H_N$ with corresponding energies $\{E_j\}$. Let
\begin{equation}
\Delta:=\min_{j\neq k}|E_j-E_k|,\quad\delta:=\min_{((j\neq j')\lor(k\neq k'))\land((j\neq k')\lor(k\neq j'))}|E_j+E_k-E_{j'}-E_{k'}|
\end{equation}
so that $\delta\le\Delta$. Let $|\circ\rangle,|\bullet\rangle$ be eigenstates of $\sigma^x_{N+1}$ with eigenvalues $\pm1$. The eigenstates of $H_{\rm unp}$ are $\{|j\circ\rangle,|j\bullet\rangle\}$ with corresponding energies $\{E_j\pm h_N\}$. We choose $h_N,J_N$ such that
\begin{equation} \label{eq:cond}
    0<J_N\ll h_N\ll\min\big\{\delta,\Delta\min_{j\neq k}\big|\langle j|\hat\sigma^z_N|j\rangle-\langle k|\hat\sigma^z_N|k\rangle\big|\big\}.
\end{equation}
The induction hypotheses imply that the most right-hand side of (\ref{eq:cond}) is positive. The condition (\ref{eq:cond}) implies that the energy gap (minimum difference between adjacent eigenvalues) of $H_{\rm unp}$ is $2h_N$. Let $\{|j\circ),|j\bullet)\}$ be eigenstates of $H_{N+1}:=H_{\rm unp}+H_{\rm per}$ with corresponding energies $\{E_{j\circ},E_{j\bullet}\}$.

{\bf Proof of induction hypothesis 1 for $H_{N+1}$.} Using second-order non-degenerate perturbation theory,
\begin{align}
E_{j\circ}&=E_j+h_N+\langle j\circ|H_{\rm per}|j\circ\rangle+\sum_{k\neq j}\frac{\big|\langle k\circ|H_{\rm per}|j\circ\rangle\big|^2}{E_j-E_k}+\sum_k\frac{\big|\langle k\bullet|H_{\rm per}|j\circ\rangle\big|^2}{E_j-E_k+2h_N}+O(J_N^3)\nonumber\\
&=E_j+h_N+\frac{J_N^2(2+\langle j|\hat\sigma^z_N|j\rangle)^2}{2h_N}+\sum_{k\neq j}\frac{J_N^2\big|\langle k|\hat\sigma^z_N|j\rangle\big|^2}{E_j-E_k+2h_N}+O(J_N^3)\nonumber\\
&=E_j+h_N+\frac{J_N^2(2+\langle j|\hat\sigma^z_N|j\rangle)^2}{2h_N}+O(J_N^2/\Delta),
\end{align}
where we used
\begin{equation}
    \sum_{k\neq j}\frac{\big|\langle k|\hat\sigma^z_N|j\rangle\big|^2}{|E_j-E_k+2h_N|}\le\sum_k\frac{\big|\langle k|\hat\sigma^z_N|j\rangle\big|^2}{\Delta-2h_N}=\frac{1}{\Delta-2h_N}=O(1/\Delta).
\end{equation}
Similarly,
\begin{equation}
E_{j\bullet}=E_j-h_N-\frac{J_N^2(2+\langle j|\hat\sigma^z_N|j\rangle)^2}{2h_N}+O(J_N^2/\Delta).
\end{equation}
Suppose that
\begin{equation}
    E_{j\alpha}+E_{k\beta}=E_{j'\alpha'}+E_{k'\beta'},\quad\alpha,\beta,\alpha',\beta'\in\{\circ,\bullet\}.
\end{equation}
Assume without loss of generality that $j\le k$ and $j'\le k'$. Since $h_N\ll\delta$, we have $j=j'$ and $k=k'$. Since $J_N\ll h_N$, it suffices to exclude the possibilities that (1) $j\neq k$ and $\alpha=\beta'=\circ$ and $\beta=\alpha'=\bullet$ or (2) $j\neq k$ and $\alpha=\beta'=\bullet$ and $\beta=\alpha'=\circ$. In both cases,
\begin{multline}
    \big|J_N^2(2+\langle j|\hat\sigma^z_N|j\rangle)^2/h_N-J_N^2(2+\langle k|\hat\sigma^z_N|k\rangle)^2/h_N\big|=O(J_N^2/\Delta)\\
    \implies\big|\langle j|\hat\sigma^z_N|j\rangle-\langle k|\hat\sigma^z_N|k\rangle\big|=O(h_N/\Delta).
\end{multline}
This is impossible for our choice of $h_N$ (\ref{eq:cond}).

{\bf Proof of induction hypothesis 2 for $H_{N+1}$.} Using first-order non-degenerate perturbation theory,
\begin{multline}
    |j\circ)=|j\circ\rangle+\sum_{k\neq j}\frac{\langle k\circ|H_{\rm per}|j\circ\rangle}{E_j-E_k}|k\circ\rangle+\sum_k\frac{\langle k\bullet|H_{\rm per}|j\circ\rangle}{E_j-E_k+2h_N}|k\bullet\rangle+O(J_N^2)\\
    =|j\circ\rangle+\frac{J_N(2+\langle j|\hat\sigma^z_N|j\rangle)}{2h_N}|j\bullet\rangle+\sum_{k\neq j}\frac{J_N\langle k|\hat\sigma^z_N|j\rangle}{E_j-E_k+2h_N}|k\bullet\rangle+O(J_N^2)
\end{multline}
so that
\begin{equation}
    (j\circ|\hat\sigma^z_{N+1}|j\circ)=J_N(2+\langle j|\hat\sigma^z_N|j\rangle)/h_N+O(J_N^2).
\end{equation}
Similarly,
\begin{equation}
    (j\bullet|\hat\sigma^z_{N+1}|j\bullet)=-J_N(2+\langle j|\hat\sigma^z_N|j\rangle)/h_N+O(J_N^2).
\end{equation}
Induction hypothesis 2 for $H_N$ implies that $\{(j\circ|\hat\sigma^z_{N+1}|j\circ),(j\bullet|\hat\sigma^z_{N+1}|j\bullet)\}_{j=1}^{2^N}$ is a set of pairwise distinct numbers.
\end{proof}

\section{Extensions of Theorem \ref{thm:generic}} \label{app:extend}

\subsection{Proof of Corollary \ref{cor:mbl}}

Following the proof of Theorem \ref{thm:generic} in Appendix \ref{app:generic}, it suffices to find a particular $\gamma\in\mathbb R^{3N-1}$ (not necessarily $\gamma\in[-1,1]^{\times(3N-1)}$) such that the spectrum of $H_h^{XXZ}(\gamma)$ has non-degenerate gaps. Let
\begin{equation}
    H'_N:=\Lambda H_N+\sum_{j=1}^{N-1}(\hat\sigma_j^x\hat\sigma_{j+1}^x+\hat\sigma_j^y\hat\sigma_{j+1}^y),
\end{equation}
where $H_N$ is given by Eq. (\ref{eq:HN}). Since the spectrum of $H_N$ has non-degenerate gaps (Lemma \ref{l:HN}), that of $H'_N$ also has non-degenerate gaps for sufficiently large $\Lambda>0$.

\subsection{Higher spatial dimensions}

To prove an analogue of Theorem \ref{thm:generic} on a higher-dimensional lattice $L$, it suffices to construct a particular Hamiltonian on $L$ with nearest-neighbor interactions such that its spectrum has non-degenerate gaps. To this end, we simply embed the spin chain model (\ref{eq:HN}) into $L$. Figure \ref{lattice} illustrates the embedding for $L$ being a two-dimensional square lattice.

\begin{figure}
\centering
\begin{tikzpicture}
\draw[help lines](-.2,-.2) grid (4.2,4.2);
\draw[ultra thick](0,0)--(0,4)--(1,4)--(1,0)--(2,0)--(2,4)--(3,4)--(3,0)--(4,0)--(4,4);
\end{tikzpicture}
\caption{Embedding a spin chain (thick line) into a two-dimensional square lattice (grid) such that adjacent sites in the spin chain remain adjacent in the square lattice.}
\label{lattice}
\end{figure}
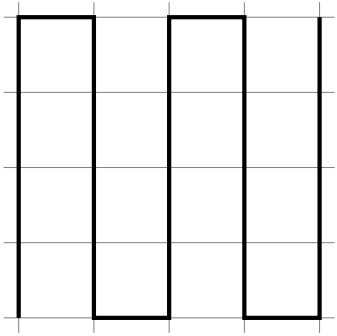

\subsection{Qudit systems} \label{app:qudit}

To extend Theorem \ref{thm:generic} to a qudit chain with $d\ge3$, it suffices to prove an analogue of Lemma \ref{l:HN}. Let $-1\le e_1<e_2<\cdots<e_d\le1$ be $d$ constants such that the spectrum of the diagonal matrix $\hat X:=\diag(e_1,e_2,\ldots,e_d)$ has non-degenerate gaps. Let
\begin{equation}
\hat Z:=\bigoplus_{j=1}^{d/2}\begin{pmatrix}
0 & 1 \\
1 & 0
\end{pmatrix},
\quad\hat Z:=\bigoplus_{j=1}^{(d-1)/2}
\begin{pmatrix}
0 & 1 \\
1 & 0
\end{pmatrix}
\oplus
\begin{pmatrix}
0
\end{pmatrix}
\end{equation}
if $d$ is even/odd, respectively, so that the expectation value of $\hat Z$ in any eigenstate of $\hat X$ is $0$. Let $\Lambda>1$ be a constant such that
\begin{equation}
    \Lambda>\max_{1\le j<k\le d/2}\frac{e_{2j}-e_{2j-1}+e_{2k}-e_{2k-1}}{|e_{2j}-e_{2j-1}-e_{2k}+e_{2k-1}|}.
\end{equation}

Consider a chain of $N$ qudits. Let $\hat X_j,\hat Z_j$ be the $\hat X,\hat Z$ operators for the spin at position $j$. It is trivial to construct a Hamiltonian $H_1$ acting only on the first spin such that
\begin{enumerate}
    \item The spectrum of $H_1$ has non-degenerate gaps.
    \item $\{\langle j|\hat Z_1|j\rangle\}$ is a set of pairwise distinct numbers, where $\{|j\rangle\}_{j=0}^{d-1}$ are eigenstates of $H_1$.
\end{enumerate}

\begin{corollary}
For any positive integer $N$, there exists $(h_1,J_1,h_2,J_2,\ldots,h_{N-1},J_{N-1})\in\mathbb R^{2(N-1)}$ such that the spectrum of 
\begin{equation}
    H_N=H_1+\sum_{j=1}^{N-1}\big(h_j\hat X_{j+1}+J_j(\Lambda+\hat Z_j)\hat Z_{j+1}\big)
\end{equation}
has non-degenerate gaps.
\end{corollary}

This corollary can be proved in almost the same way as Theorem \ref{thm:generic}.

\section{Proof of Theorem \ref{thm:tight}}

\begin{lemma} \label{l:ent}
Let $\{|j\rangle\}_{j=0}^{d-1}$ be an orthonormal basis of $\mathbb C^d$. For a Haar-random state $|\Psi'\rangle$,
\begin{equation}
    -\E_{\Psi'}\sum_{j=0}^{d-1}p_j\ln p_j=\sum_{k=2}^d\frac1k,\quad p_j:=\big|\langle j|\Psi'\rangle\big|^2.
\end{equation}
\end{lemma}

\begin{proof}
\begin{equation}
    \E_{\Psi'}\sum_{j=0}^{d-1}p_j\ln p_j=d\E_{\Psi'}p_0\ln p_0=d\int_0^1xf(x)\ln x\,\mathrm dx=-\sum_{k=2}^d\frac1k,
\end{equation}
where $f(x)$ is given by Eq. (\ref{eq:fx}).
\end{proof}

Let
\begin{equation}
\sigma_A:=\lim_{\tau\to+\infty}\frac1\tau\int_0^\tau\rho_A(t)\,\mathrm dt
\end{equation}
be the infinite time average and $\sigma_j:=\tr_{A\setminus\{j\}}\sigma_A$ be the reduced density matrix of the spin at position $j\in A$. Using the concavity and subadditivity \cite{AL70} of the von Neumann entropy,
\begin{equation} \label{eq:conca}
    \lim_{\tau\to+\infty}\frac1\tau\int_0^\tau S\big(\rho_A(t)\big)\,\mathrm dt\le S(\sigma_A)\le\sum_{j\in A}S(\sigma_j).
\end{equation}
Let $x\doteq y$ denote that $x-y$ is infinitesimal. Recall that $H^{\rm mbl}\doteq H^{\rm loc}$, which is diagonal in the computational basis and whose spectrum is non-degenerate. For a Haar-random product state $|\Psi\rangle=\bigotimes_{j=1}^N|\Psi_j\rangle$, it is not difficult to see that
\begin{equation} 
    \sigma_j\doteq\diag(p_{j,0},p_{j,1},\ldots,p_{j,d-1}),\quad p_{j,k}:=\big|{}_j\langle k|\Psi_j\rangle\big|^2,
\end{equation}
where $\{|k\rangle_j\}_{k=0}^{d-1}$ is the computational basis for the spin at position $j$. Hence,
\begin{equation} \label{eq:subadd}
    \E_\Psi\sum_{j\in A}S(\sigma_j)\doteq-\sum_{j \in A}\E_{\Psi_j}\sum_{k=0}^{d-1}p_{j,k}\ln p_{j,k}.
\end{equation}
Theorem \ref{thm:tight} follows by combining Lemma \ref{l:ent} and (\ref{eq:conca}), (\ref{eq:subadd}).

\bibliographystyle{abbrv}
\bibliography{main}

\end{document}